\documentclass{tCON2e}
\usepackage{epstopdf}
\usepackage{subfigure}

\theoremstyle{plain}
\newtheorem{theorem}{Theorem}

\newtheorem{lemma}{Lemma}

\theoremstyle{definition}

\theoremstyle{remark}

\begin{document}



\title{Contraction based stabilization of nonlinear singularly perturbed systems and application to high gain feedback}

\author{
\name{M M Rayguru\textsuperscript{a}$^{\ast}$\thanks{$^\ast$Corresponding author. Email: mmrayguru87@gmail.com}
and I N Kar\textsuperscript{a}}
\affil{\textsuperscript{a} Electrical Engineering Department, IIT Delhi, New Delhi, 110016, India;
}}

\maketitle

\begin{abstract}
Recent development of contraction theory  based analysis of singularly perturbed system has opened the door for inspecting differential  behavior of multi time-scale systems. In this paper a contraction theory based framework is proposed for stabilization of singularly perturbed systems. The primary objective is to design a feedback controller to achieve bounded tracking error for both standard and non-standard singularly perturbed systems. This framework provides relaxation over traditional quadratic Lyapunov based method as there is no need to satisfy interconnection conditions during controller design algorithm. Moreover, the stability bound does not depend on smallness  of singularly perturbed parameter. Combined with high gain scaling, the proposed technique is shown to assure contraction of approximate feedback linearizable systems. These findings extend the class of nonlinear systems which can be made contracting.
\end{abstract}


\section{Introduction}
Multi time-scale modeling and study of singularly perturbed systems find application in model order reduction, optimal control, stochastic filtering and composite control etc.  \cite{Kokotovic1986}. A two time scale singularly perturbed system consists of an interconnection of two dynamical systems referred as slow and fast subsystems. Generally, we refer a singularly perturbed systems as standard models if there exists a unique root for the fast subsystem when the  perturbation parameter goes to zero \cite{sab1}. Whereas in nonstandard models, fast system will have multiple roots or without any root \cite{nar1}.\par
Quadratic Lyapunov function has been effectively utilized for stability analysis and controller design for singularly perturbed systems \cite{sab2,son1}. For the purpose of analysis the overall system is separated into two reduced order models by setting the perturbation parameter to zero.  Stability of each reduced system is investigated by selection of two appropriate quadratic Lyapunov functions. Subsequently the convex sum of these two functions (Composite Lyapunov Function) is employed to assure stability of the overall system. The resulting stability bounds are valid for certain range of the perturbation parameter depending on the interconnection conditions satisfied by the  Lyapunov functions \cite{sab1}. In addition to solving regulation problems quadratic Lyapunov functions have been efficient on examining closed loop stability of output feedback controllers, high gain feedback \cite{son1}, dynamic surface control  etc. However the composite Lyapunov approach encounters complication in analyzing nonstandard models. Indirect manifold construction with a modified composite control law is used for stabilization of nonstandard problems  \cite{nar1} and references therein. Nevertheless it is not always easy to search for two quadratic Lyapunov functions which should satisfy all the interconnection conditions  and moreover presence of uncertainties further complicates the search process. The condition for stability is a sufficient one and hence there is no guarantee of stability beyond the critical value of perturbation parameter.\par
Recently a differential form of stability analysis, namely contraction theory is proposed in \cite{lohmiller1}. In agreement with this hypothesis all the trajectories of a contracting dynamical system exponentially converge towards each other irrespective of their initial condition \cite{sontag1, angeli}. The region in the state space is called contraction region, if every trajectory starting inside the region will converge towards each other \cite{parrilo}. Contraction framework does not necessitates the presence of an attractor a priori, however a contracting autonomous system indirectly assures the presence of an equilibrium point \cite{lohmiller1}. The exponential convergence property is inherently robust to bounded disturbances and hence easier to deal with uncertainties in system model \cite{sontag1}. These interesting properties are utilized for analysis of mechanical systems \cite{slotine7}, stability of networks, observer design \cite{juo3}, synchronization \cite{slotine6}, Kalman filter, frequency estimator design \cite{sharma3}, backstepping controller synthesis \cite{zamani1, juo1} etc.\par
Moreover contraction framework is extended to analyze singularly perturbed systems \cite{delv} and its application to retroactive attenuation in biological systems \cite{delv2}. These results are employed for stabilization of approximate feedback linearizable systems. Partial contraction analysis and robustness of contraction property is exploited to derive new stability bounds for singularly perturbed nonlinear systems in these works. The procedure is recursive and can be extended to three or multi time scale systems. The stability bounds obtained hold for a broad range of perturbation parameter rather than a small range found in quadratic Lyapunov based methods. Therefore contraction framework based analysis of singular perturbed system provides less conservative bounds compared to conventional Lyapunov methods. \\
In this paper we will show, how contraction theory tools can be adopted for stabilization problems in standard and nonstandard singularly- perturbed systems. The use of contraction tools completely circumvents the need of interconnection conditions and guarantees convergence behavior for a broad range of perturbation parameter. Thus the proposed method provide guaranteed stability bounds for a wide range of perturbation parameter which is difficult to obtain using quadratic Lyapunov based formalism. The  design procedure is also extended to high gain scaling based control law for a class of approximate feedback linearizable systems. For these cases, parameter selection for controller and convergence analysis is guaranteed in the formalism of contraction theory. The method presented in this paper will complement the composite controller design approach when searching for quadratic Lyapunov functions satisfying all the interconnection condition becomes difficult.\\ The paper is organized as follows. The motivation and problem formulation is discussed in the first section followed by some discussion on contraction theory. Stabilization of standard and nonstandard singularly perturbed systems  are derived next. Application of these results to high gain feedback controller design for approximate feedback linearizable systems is presented in subsequent section. Finally in the last section we present simulation results for some examples.\\

Throughout this paper, we adopt the following notations and symbols. $B_x, B_z$ denote compact subsets, $R^m$ denotes a m-dimensional real vector space. For real vectors $v$, $||v||$ denotes the Euclidean norm and for real matrix $||E||$ denotes induced matrix norm. A metric $\Theta$ denotes a symmetric positive definite matrix and $I_n$ is an $n \times n$ identity matrix. \par

\section{Motivation and Problem Formulation}

\subsection{Motivation} The composite controller design approach for standard singularly perturbed systems is discussed through an example. The system is described as 
\begin{subequations}\label{mot1}
 \begin{equation}\label{reg1}
 \dot{x}=f(x,z)=xz^3, x \in B_x=[-1, 1]
 \end{equation}
 \begin{equation}\label{mote1}
 \mu\dot{z}= g(z,x,u)=z +u, z \in B_z=[-1/2, 1/2]
 \end{equation}
 \end{subequations}
\quad where $\mu$ is small positive number less than one. The design goal is to stabilize the system around the origin using composite Lyapunov function based technique discussed in \cite{sab1,nar1}. The composite control law $u$ consists of a slow ($u_a$) and a fast component ($u_b$) which are selected in a recursive manner. We divide the procedure into three distinct steps mentioned below. \\
\emph{Stabilization of Reduced Slow System:} The slow component is selected under the assumption that, there exists a slow manifold for the model and all the fast states $z$ have converged to this manifold. Equating $\mu=0$ in \eqref{mot1}, the root of the fast subsystem or the slow manifold is given by,
\[z=z_{ds}=h(x,u)=-u.\]   Note that, the fast component $u_b$ of the control law vanishes when $z\rightarrow h(x,u)$. The slow component $u_a$ has to be selected in such a way that the reduced system will be stable. 
A control law $u_a=x^{4/3}$ will stabilize the reduced slow system $\dot{x}=f(x,h(x,u))=-x^2$ around origin. For this choice of $u_a$, a candidate Lyapunov function $V(x)=\frac{1}{2}x^6$ will satisfy the following inequality\\
\begin{equation}
\dot{V} \leq \alpha_1 \psi(x)^2
\end{equation}
where $\alpha_1=1,\psi(x)=||x||^5$.\\
\emph{Stabilization of Boundary Layer System:} The boundary layer system for \eqref{mot1} is written as:
\[\dot{y}=y+u_b\] where $y=z-h(x,u)$. A choice of $u_b=-3(z-x^{4/3})$ will achieve stability of boundary layer system using Lyapunov function $W=\frac{1}{2}(z-z_{ds})^2$. The derivative of $W$ along the trajectories of boundary layer system will follow
\begin{equation}
\dot{W} \leq \alpha_2 \phi(z-z_{ds})^2
\end{equation}
where $\alpha_2=2$ and $\phi(z-z_{ds})=||z-z_{ds}||$.\\
\emph{Interconnection Conditions:} Overall system stability of \eqref{mot1} is examined by selecting a composite Lyapunov function which is a convex sum of $V$ and $W$. Moreover to assure asymptotic stability  of  \eqref{mot1}, $V(x)$ and $W(y)$ must satisfy the following interconnection conditions.
\begin{equation}\label{intc}
\begin{split}
& \frac{\partial{V}}{\partial{x}}(f(x,y+h(x))-f(x,h(x))) \leq \beta_1\phi(x)\psi(y)\\
& \frac{\partial{W}}{\partial{y}}(f(x,y+h(x))) \leq \beta_2\phi(x)\psi(y)+\beta_3\psi(y)^2
\end{split}
\end{equation}
In the region ($B_x \times B_z$) given in \eqref{mot1}, the choice of scalars $\beta_1=7/4,\beta_2=4/3, \beta_3=7/3$ will satisfy the interconnection conditions \eqref{intc}. The maximum value of perturbation parameter for which the system \eqref{mot1} is asymptotically stable depends on the interconnection conditions and selection of composite Lyapunov function. The maximum bound of perturbation parameter $(\mu= 0.4246)$ can be achieved by selecting a composite Lyapunov function $V_1=(1-d) V+d W$ where $d=21/47$.\par
From above discussion it can be concluded that, composite control approach provides an elegant and step by step design for stabilization problems. The idea is to reduce the complexity of the overall system by converting it into two reduced order models and thereafter sensibly selecting two components of the control law for two reduced systems. The stability of the overall system hinges on the selection of Lyapunov functions $V$, $W$ and interconnection conditions. Moreover the condition for stability is a sufficient one and the system \eqref{mot1} may be stable beyond the maximum predicted range of $\mu$. The closed loop system for \eqref{mot1} is simulated for a choice of $\mu=0.5$ and the result is shown in figure 1.
\begin{figure}[htbp]
     \centering
      \includegraphics[width=3.25in,height=2.20in]{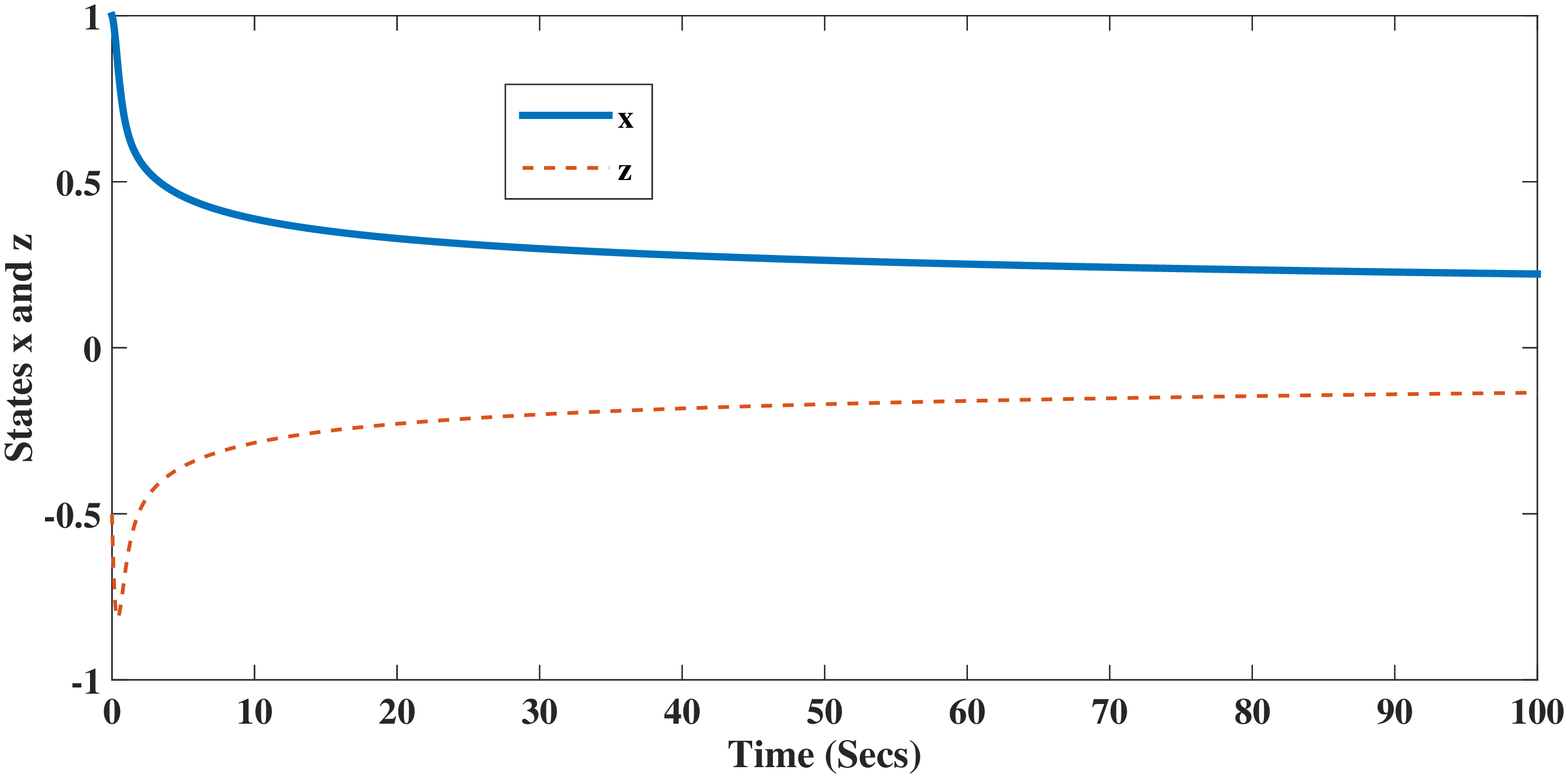}
     \caption{ Closed Loop System Response for $\mu=0.5$}
      \label{standard1}
\end{figure} 
It is hard to conclude asymptotic stability of \eqref{mot1} from the figure but it is clear that the closed loop system trajectories are converging in the neighborhood of origin. This behavior of trajectories can not be concluded from composite Lyapunov function approach. A proportionate stability bound with respect to $\mu$ may give some relaxation to control design in many practical cases where a uniform ultimate boundedness is the design requirement. In high gain observer based output feedback design  or high gain feedback based control designs \cite{son1}, closed loop system stability is guaranteed only for a very small range of perturbation parameter. These restriction can be relaxed if the stability of the closed loop system can be assured for a wide range of perturbation parameter.\\
Also in presence of systems uncertainties, searching for Lyapunov functions satisfying all the interconnection condition is a difficult task. In this paper we are proposing a contraction theory framework for the stabilization of singularly perturbed system addressing these stability issues. The control algorithm proposed in this paper retains the idea of reducing the model order by equating the perturbation parameter to zero. However there is no need to search for Lyapunov functions satisfying interconnection conditions.\par 
\newpage
\subsection{Problem Formulation} In this paper, we investigate the stabilization of \eqref{mot3} in contraction theory framework. Consider the standard singularly perturbed systems described as: 
 \begin{subequations}\label{mot3}
 \begin{equation}\label{sys1s}
 \dot{\bm x}=\it \bf f (\bm x, \bm z, u)
 \end{equation}
 \begin{equation}\label{sys1sf}
 \mu\dot{\bm z}=\it \bf g (\bm x, \bm z,\mu, u)
 \end{equation}
 \end{subequations}
where $\bm x \in \rm \mathbb B_x \subset \rm \mathbb R^n, \bm z \in \rm \mathbb B_z \subset \rm \mathbb R^m$, $u \in R$ and $\mu \in [0, 1]$. The functions $f(.)$, $g(.)$ are assumed to be smooth and Lipschitz in their arguments. We solve the following problems.\\
  $\rm I)$ Design a control law  $\it u=\it u_1(\it \bf x)+\it u_2(\it \bf x, \it \bf z)$ such that the trajectories of the overall system converge to an ultimate bound irrespective of the magnitude of the perturbation parameter. \\
 $\rm II)$ Further we show that the proposed control law is robust against bounded disturbances and derive the convergence bounds depending on the magnitude of the disturbance term.\\
  $\rm III)$ Investigate the special cases for which the trajectories of \eqref{mot3} exponentially converge to an equilibrium.\\
  $\rm IV)$ Explore the design steps required for stabilization of nonstandard singularly perturbed models and derive the convergence bounds.\\
  $\rm V)$ Exploit the proposed approach to design high gain feedback controllers for the stabilization of approximate feedback linearizable systems.
\section{Prerequisites From Contraction Theory}

A system of the form $\dot{\bm x}=\it \bf f (\bm x, t)$ is said to be contracting if all trajectories starting inside some region in state space will converge to each other exponentially somehow forgetting their initial conditions or disturbances \cite{lohmiller1}. Existence of such region is sufficient  for guarantying contraction behavior in a dynamical system.  A region in the state space is called a contraction region for the system, if the following inequality is satisfied for $\forall t >0$
\begin{equation}\label{cont1}
\rm \bf F=\it (\dot{\rm \bf \Theta}+\rm \bf \Theta\frac{\upartial{\it \bf f}}{\upartial{\it \bf x}})\rm \bf \Theta^{-1} \leq -\it \lambda \rm \bf  I
\end{equation}
where $\rm \bf \Theta$ is a nonsingular metric, $\it \lambda>0$ is a positive constant referred as contraction rate, $\rm \bf F$ is defined as generalized Jacobian. This condition can also be expressed in an inequality form as:
\begin{equation}\label{econt}
(\dot{\rm \bf M}+\rm \bf M\frac{\upartial{\it \bf f}}{\upartial{\it \bf x}}+\frac{\upartial{\it \bf f}}{\upartial{\it \bf x}}^{T}\rm \bf M) \leq -2\it \lambda \rm \bf M
\end{equation}
where $\rm \bf M(\it \bf x,t)=\rm \bf \Theta^T\rm \bf \Theta$ is an uniformly positive definite matrix. The system is said to be contracting in a metric $\rm \bf M$ with a rate $\it \lambda$, when inequality \eqref{econt} is satisfied.  When this inequality is moderated into a negative semidefinite condition, the system is said to be semi-contracting.  Some important results and observations from previous literatures are outlined in the form of following lemmas whose proofs can be found in \cite{lohmiller1,parrilo,sontag1}.\par
\begin{lemma}\label{lem1} 
Suppose an autonomous system $\dot{\it \bf x}=\it \bf f(\it \bf x)$ is globally contracting with a nonsingular metric $\rm \bf \Theta(\it \bf x)$ then all the trajectories of this system will converge to an unique equilibrium point. $\Diamond$ 
\end{lemma}
Contraction of a dynamical system points to the local behavior of differential displacements of its trajectories. From the local analysis of the virtual displacements a fair idea about the global behavior can be drawn for the system under consideration. Furthermore contraction provides  inherent robustness to bounded uncertainties affecting the system. The robustness property of contraction in a perturbed nonlinear system is summarized in the form of the following lemma. \par
\begin{lemma}\label{lem2}
Define a perturbed system of the following form.
\begin{equation} \label{pertsys}
\dot{\it \bf x_p}= \it \bf f(\it \bf x_p, t)+\it \bf d(\it \bf x_p, t)
\end{equation} 
Suppose the system $\dot{\it \bf x}=\it \bf f(\it \bf x, t)$ is contracting using a nonsingular metric $\rm \bf \Theta$ with a rate $\it \lambda$. Then the following two cases arise.\\
a) Assume the Jacobian of the perturbation term $\||\frac{\upartial{\it \bf d(\it \bf x_p, t)}}{\upartial{\it \bf x_p}}\|| \leq \it \lambda$ for $\forall t > 0$, then the perturbed system is still contracting and the trajectories of the perturbed system will exponentially converge to the trajectories of nominal(unperturbed) system. i.e\\
\begin{equation}\label{lem21}
\lim_{t \to \infty} ||\it \bf x_p(t)-x(t)|| \rightarrow 0
\end{equation}
b) If the perturbation term is bounded, then the difference between trajectories of the perturbed system and the nominal system will converge to a steady state bound given as: 
\begin{equation}\label{lem22}
\lim_{t \to \infty} ||\it \bf x_p(t)-x(t)|| \leq \frac{\it \chi \it d}{\it \lambda}
\end{equation}
where $\it \chi$ is the condition number of the metric $\rm \bf \Theta$ and $||\it \bf d(\it \bf x, t)|| \leq \it d$.
$\Diamond$
\end{lemma}
Apart from these properties contraction framework provides a very useful tool called partial contraction \cite{slotine6}, which find application in observer/filter design and synchronization problems.
\begin{lemma}\label{lem3}
A System $\dot{\it \bf x}=\it \bf f(\it \bf x, \it \bf y, t)$ is said to be partially contracting in $\it \bf x$ if an auxiliary system defined by $\dot{\it \bf z}=f(\it \bf z, \it \bf y,t)$ is contracting for any value of $\it \bf y, \forall t>0$. If the auxiliary system verifies a smooth specific property, then the trajectories of original system will verify that property exponentially. $\Diamond$
\end{lemma}

\section{Control Law Formulation}
\subsection{Controller Design For Standard Models}
For the system \eqref{mot3}, the root of the fast z-subsystem \eqref{sys1sf} is given by:
\begin{equation}
\it \bf g(\it \bf x, \it \bf z,\mu, \it u)=0 \Rightarrow \it \bf z=\it \bf h(\it \bf x,\it u_1,\mu)
\end{equation}
where $\it u_1$ is the part of the control law which remains after the fast variables $\it \bf z$ of \eqref{sys1s} has reached their steady state. Denote $\it \bf z_{ds}=\it \bf h(\it \bf x, \it u_1)$ as the root of $g(\it \bf x, \it \bf z, \rm 0, \it u)=0$ and assume it to be smooth. The reduced slow system can be expressed as: 
\begin{equation}\label{redc1}
\it \bf \dot{x}=\it \bf f(\it \bf  x, \it \bf h(\it \bf x, \it u_1), \it u_1)
\end{equation}
 Define an auxiliary system
\begin{equation}\label{viro1}
\begin{split}
&\it \bf \mu \dot{z}_{ds}=\it \bf g(\it \bf x,\it \bf z_{ds},\rm 0,\it u)+\mu\frac{\upartial{\it \bf h}}{\upartial{\it \bf x}}\dot{\it \bf x}\\
\end{split}
\end{equation}
System \eqref{viro1} can be regarded as perturbed virtual (copy/auxiliary) system for $\it \bf z$. The extra term $\it \bf g(\it \bf x,\it \bf z_{ds},\rm 0,\it u)$ is added intentionally so that, \eqref{viro1} becomes a copy of \eqref{sys1sf}. The similarity between \eqref{sys1sf} and \eqref{viro1} is obvious when $\it \bf g(.)$ does not depend on $\mu$.\\
\emph{\textbf{Remark 1}:} The advantage of using contraction framework is that, we do not need any error dynamics (boundary layer system) for our analysis. The equations \eqref{redc1} and \eqref{viro1} play analogous role as the reduced system and boundary layer system in conventional Lyapunov based design. Our work is inspired from \cite{delv} where contraction analysis of singularly perturbed system is the main goal. However we are more interested to solve stabilization problem of singularly perturbed system and approximate feedback linearizable systems.
The main result of this section is stated in the following theorem.
\begin{theorem}\label{The1} 
Assume the following statements are true. \\
i) There exists a smooth function $\it u_1(\it \bf x)$ and a metric $\rm \Theta_x$ such that the reduced system \eqref{redc1} is contracting.\\
ii) There exists a control law $\it u_2(\it \bf x, \it \bf z)$ satisfying $||\it u_2|| \leq \it d_2||\it \bf z-\it \bf h(\it \bf x, \it u_1)||$ in $(\rm \bf B_x \times \rm \bf B_z)$ and a metric $\rm \Theta_z$such that the system $\mu\it \bf \dot{z}=\it \bf g(\it \bf x,\it \bf z,\mu,\it u_1+u_2)$ is partially contracting in $\it \bf z$.\\ 
iii) $||\frac{\upartial{\it \bf z_{ds}}}{\upartial{\it \bf x}}\it \bf f(\it \bf x,\it \bf z)|| \leq \it d_1$ in $(\rm \bf B_x \times \rm \bf B_z)$ where $\it d_1$ is a positive constant.\\
Then there exists a control law $\it u=\it u_1(\it \bf x)+\it u_2(\it \bf x, \it \bf z)$  such that \eqref{mot3} is contracting in $(\rm \bf B_x \times \rm \bf B_z)$ and it's trajectories follow the bounds given in \eqref{bd1} and \eqref{bd2}.\\
\end{theorem}
\begin{proof}
Adding and subtracting $\it \bf g(\it \bf x,\it \bf z_{ds},\rm \mu,\it u)$ in \eqref{viro1}, we get
\begin{equation}\label{virz2}
\Rightarrow \it \mu \bf \dot{z}_{ds}=\it \bf g(\it \bf x,\it \bf z_{ds},\rm \mu,\it u)+\mu\frac{\upartial{\it \bf h}}{\upartial{\it \bf x}}\dot{\it \bf x} +\it \bf g(\it \bf x,\it \bf z_{ds},\rm 0,\it u)-\it \bf g(\it \bf x,\it \bf z_{ds},\rm \mu,\it u)
\end{equation}
The virtual system \eqref{virz2} appears to be a perturbed form of the fast subsystem \eqref{sys1sf}. Using third condition of theorem \ref{The1} states, there exist a control input $u_2$ such that the nominal part $\it \bf \dot{z}_{ds}=\it \bf g(\it \bf x,\it \bf z_{ds},\rm \mu,\it u)$ is partially contracting in $\it \bf z$. From the Lipschitz assumption for $\it \bf g(.)$ in $\mu$, \\
\[||\mu\frac{\upartial{\it \bf h}}{\upartial{\it \bf x}}\dot{\it \bf x} +\it \bf g(\it \bf x,\it \bf z_{ds},\rm 0,\it u)-\it \bf g(\it \bf x,\it \bf z_{ds},\rm \mu,\it u)|| \leq \mu (\it L_1+ \it d_1).\]
where $\it L_1$ is Lipschitz constant. Using lemma \ref{lem2}, by exploiting the robustness property of contraction we can derive the error bound between trajectories of the original fast subsystem \eqref{sys1sf} and the desired slow manifold whose dynamics is described by \eqref{virz2}. The error bound can be expressed as,
\begin{equation}\label{bd1}
\begin{split}
&||\it \bf z(t)-\it \bf z_{ds}(t)|| \leq \chi_z e^{-\lambda_{z}t/\mu}||\it \bf z(\rm 0)-\it \bf z_{ds}(\rm 0)||
+ \frac{\mu \chi_z}{\lambda_{z}}(\it d_1 +\it  L_1)
\end{split}
\end{equation}
where $\chi_z$ is the condition number of contraction metric $\rm \bf \Theta_z$ and $\lambda_{z}$ is the contraction rate. Similar to the above argument a virtual system of x-subsystem  \eqref{sys1s} can be defined as:
\begin{equation}\label{vir2}
\begin{split}
&\it \bf \dot{x}_r=\it \bf f(\it \bf x_r,\it \bf h(\it \bf x_r, \it u_1),\it u_1)+\it \bf f(\it \bf x_r,\it \bf z, \it u)-\it \bf f(\it \bf x_r,\it \bf h(\it \bf x_r, \it u_1),\it u_1)\\
&\leq \it \bf f(\it \bf x_r,\it \bf h(\it \bf x_r, \it u_1),\it u_1)+\it \bf f(\it \bf x_r,\it \bf z, \it u)-\it \bf f(\it \bf x_r,\it \bf z, \it u_1)+\it \bf f(\it \bf x_r,\it \bf z, \it u_1)-\it \bf f(\it \bf x_r,\it \bf h(\it \bf x_r, \it u_1),\it u_1)\\
& \leq \it \bf f(\it \bf x_r,\it \bf h(\it \bf x_r, \it u_1),\it u_1) + \it L_u||\it u_2|| + \it L_2||\it \bf z-\it \bf z_{ds}||\\
& \leq \it \bf f(\it \bf x_r,\it \bf h(\it \bf x_r, \it u_1),\it u_1) + \it L_ud_2||\it \bf z-\it \bf z_{ds}|| + \it L_2||\it \bf z-\it \bf z_{ds}||
\end{split}
\end{equation}  
where $\it L_u$ and $\it L_2$ are Lipschitz constants of $\it \bf g(.)$ for $\it u$ and $\it \bf z$ respectively. The bound $||\it \bf z-\it \bf z_{ds}||$ is known from \eqref{bd1} and $\it L_u||\it u_2|| \rightarrow \rm 0$ exponentially as $\it \bf z\rightarrow \it \bf z_{ds}$.  Employing lemma \ref{lem2}, $||(\it \bf x-x_r||$ decreases exponentially and satisfies the following bound.
\begin{equation}\label{xbnd}
\begin{split}
& ||\it \bf x(t)- x_r (t)|| \leq \chi_x||\it \bf x(0)- x_r (0)||\rm e^{-\lambda_x t}+ \mu ( \it C_1(\rm e^{-\lambda_x t} - e^{\frac{-\lambda_{z}}{\mu} t} ) + \it C_2(1-\rm e^{-\lambda_x t}))  \\ 
\end{split}
\end{equation}
where $\chi_x$ is the condition number of contraction metric for the reduced system and $\lambda_{x}$ is contraction rate. The constants $\it C_1$ and $\it C_2$ are given by,
\[ \it C_1 = \frac{\chi_x \chi_z(\it L_2+L_ud_2) ||\it \bf z(0)-z_{ds}(0)|}{\lambda_z - \mu \lambda_x}, 
 \it C_2 = \frac{\chi_x \chi_z(\it L_2+L_ud_2)  (\it d_1+L_1)}{\lambda_z\lambda_x}\].
 The trajectories of \eqref{sys1s} exponentially converges to the following bound.
\begin{equation}\label{bd2}
\lim_{t \to \infty} ||\it \bf x(t)-x_r(t)|| \leq \mu\frac{\chi_x \chi_z(\it L_2+L_ud_2)  (\it d_1+L_1)}{\lambda_z\lambda_x}
\end{equation}
\end{proof}
\emph{\textbf{Remark 2}:} The fourth condition given in theorem 1 appears to be a conservative condition but as the variable $x$ is varying slowly for auxiliary system  \eqref{virz2}, this is not actually very restrictive one. If  $\it \bf x$ and $\it \bf z$ evolve inside bounded regions $(\rm \bf B_x \times \rm \bf B_z)$, then this condition is reasonable. So even if global stability can not be estabilished, semi-global stability can be achived using these results. This condition is essential since no interconnection condition has to be satisfied. \\
\emph{\textbf{Remark 3}:} The stability achieved is not asymptotic in nature, rather an ultimate bound is established which depends on the perturbation parameter. The error bounds \eqref{bd1} and \eqref{bd2} is valid for all $\mu \in [0,1]$ so the proposed controller can give a relative stability result depending on $\mu$.\\
\subsection{Robustness Issues}
Suppose the z sub-system \eqref{sys1sf} is disturbed by a uncertainty $\it \bf d(\it \bf x, z, \mu)$, which can be written as:
 \begin{equation}\label{fstp}
 \mu\it \bf \dot{z}_p=\it \bf g(\it \bf x, \it \bf z_p,\mu,\it u)+\it \bf d(\it \bf x, z, \mu) 
 \end{equation}
 where $||\it \bf d(.)||\leq \it d_b$. From theorem 1, the unperturbed part of \eqref{fstp} is partially contracting in $\it \bf z_p$. Following bound can be established using lemma 2.
 \begin{equation}
 ||\it \bf z(t)-z_{p}(t)|| \leq \chi_z \rm e^{-\lambda_{z}t}||\it \bf z(0)-z_{p}(0)|| + \frac{\mu \chi_z \it d_b}{\lambda_{z}}\\
 \end{equation}
 From triangle inequality,
 \begin{equation}\label{vir1}
||\it \bf z_p-z_{ds}|| \leq ||\it \bf z_p-z||+||\it \bf z-z_{ds}||
\end{equation}
Assuming the initial condition for $\it \bf z$ and $\it \bf z_p$ are same,
the bound for $||\it \bf (z_{p}-z_{ds})||$ can be reformulated as:
\begin{equation}\label{bd3}
\begin{split}
&||\it \bf z_{p}(t)-z_{ds}(t)|| \leq \chi_z \rm e^{-\lambda_{z}t}||\it \bf z(0)-z_{ds}(0)|| + \frac{\mu \chi_z}{\lambda_{z}}(\it d_1 + L_1+d_b)\\
\end{split}
\end{equation}
Replacing \eqref{bd3} with \eqref{bd2} in \eqref{vir2} we can obtain the steady state bounds for the overall system.
\subsection{Exponential Convergence}
There are certain cases when exponential convergence of the trajectories to equilibrium can be inferred rather than convergence to an ultimate bound. Assume the right hand side of the z- subsystem \eqref{sys1sf} is independent of $\mu$, which can be written in the following form.
\begin{subequations}\label{ac1}
 \begin{equation}\label{s1s}
 \dot{\bm x}=\it \bf f (\bm x, \bm z, \it u)
 \end{equation}
 \begin{equation}\label{s1sf}
 \mu\dot{\bm z}=\it \bf g (\bm x, \bm z, \it u)
 \end{equation}
 \end{subequations}
where $\mu \in [0, 1]$ and $\it \bf g(\it \bf x, \it \bf z,\it u)=0$ has a root denoted by $\it \bf z_{ds}=\it \bf h(\it \bf x,\it u_1)$. The functions $f(.)$, $g(.)$ and $h(.)$ are assumed to be smooth and Lipschitz in their arguments.
\begin{theorem}\label{The2}
If the following conditions ($i-iii$)are satisfied for \eqref{ac1}, then there exists a control law $\it u$ and a small constant $\mu^*$ such that the system is contracting in $(\rm \bf B_x \times B_z)$ for all $\mu \leq \mu^*$.\\
i) There exists a smooth function $\it u_1(\it \bf x)$ and a metric $\rm \Theta_x$ such that the reduced system \eqref{redc1} is contracting.\\
ii) There exists a control law $\it u_2(\it \bf x, \it \bf z)$ satisfying $||\it u_2|| \leq \it d_2||\it \bf z-\it \bf h(\it \bf x, \it u_1)||$ in $(\rm \bf B_x \times \rm \bf B_z)$ and a metric $\rm \Theta_z$such that the system $\mu\it \bf \dot{z}=\it \bf g(\it \bf x,\it \bf z,\it u_1+u_2)$ is partially contracting in $\it \bf z$.\\ 
iii)  $||\it \bf \frac{\upartial{Q(x,z)}}{\upartial{z}}|| \leq \it d_q$ where $\it \bf Q(x,z)= \frac{\upartial{\it \bf h(x,u_1)}}{\upartial{\it \bf x}}\it \bf f(\it \bf x,\it \bf z)$ and $\it d_q$ is a positive constant .
\end{theorem}
\begin{proof}
The proof is similar to the proof of theorem \ref{The1}. 
 Define a virtual system as,
\begin{equation}\label{vi1}
\begin{split}
&\it \mu \bf \dot{z}_{ds}=\it \bf g(\it \bf x,\it \bf z_{ds},\it u)+\mu\frac{\upartial{\it \bf h}}{\upartial{\it \bf x}}\dot{\it \bf x}\\
& \Rightarrow \it \mu \bf \dot{z}_{ds}=\it \bf g(\it \bf x,\it \bf z_{ds},\it u)+\mu \it \bf Q(\it \bf x,z).
\end{split}
\end{equation}
The virtual system \eqref{vi1} can be considered as perturbed form of the fast subsystem \eqref{s1sf}. From the fourth condition of theorem \ref{The2},
\[\lVert {\mu \frac{\upartial{\it \bf Q(x,z)}}{\upartial{z}}}\rVert \leq \it \mu d_q.\]
Assume $\mu\it \bf \dot{z}=\it \bf g(\it \bf x,\it \bf z, \it u_1+u_2)$ is partially contracting in $\it \bf z$ with a rate $\lambda_{z}$. Therefore it is legitimate to assume a small constant $\mu^*$ such that,
\[\mu^*\it d_q \leq \lambda_{z}.\]
From the above argument and lemma \ref{lem2}, virtual system \eqref{vi1} is partially contracting in $\it \bf z$. Following lemma \ref{lem3} the trajectories of fast subsystem \eqref{s1sf} will follow the property of its virtual system \eqref{vi1} and hence will converge to $\it \bf z_{ds}$.
\begin{equation}\label{bnda1}
\lim_{t \to \infty} ||\it \bf z(t)-z_{ds}(t)|| \rightarrow \rm 0
\end{equation}
Rest of the proof follow the same steps as theorem \ref{The1}. Therefore from \eqref{vir2} and \eqref{xbnd}, the steady state bound for $\it \bf x -$ subsystem also converges to zero.
\end{proof}
\subsection{Controller Design for Nonstandrad Case}
 In this section we discuss controller design for nonstandard singularly perturbed models of the form \eqref{ac1}. Absence of an unique root brings about fundamental challenge in designing controllers for this class of systems. This bottleneck can be avoided, if the fast variable $z$ is treated as a virtual control variable in x-subsystem. Then a recursive design can be formulated to achieve closed loop stability of overall system. We follow a recursive indirect manifold construction approach to design control law for our purpose. The aim is to design a contracting controller in order to track a reference trajectory $\it \bf x_{r}(t)$. This contracting design procedure is divided into following steps.\\
\emph{Step 1:}  Defining an error signal  $\it \bf e(t)=x(t)-x_{r}(t)$, System \eqref{ac1} can be written as
  \begin{subequations}\label{syse1}
 \begin{equation}\label{err1}
 \it \bf \dot{e}=\it \bf f(\it \bf e+x_r, z, \it u)-\it \bf \dot{x}_r
 \end{equation}
 \begin{equation}\label{err2}
 \mu\it \bf \dot{z}=\it \bf g(\it \bf e+x_r, z, \it u)
 \end{equation}
 \end{subequations}
Selection of a control law using the same procedure as in standard model is not possible due to the absence of a unique root. To overcome complication, choose a virtual control variable $\it \bf z_{de}=\it \bf h(\it \bf e,\it \bf x_r,\it u, \it \bf \dot{x}_r)$ such that the reduced slow system given by:
\begin{equation}\label{reds}
\it \bf\dot{e}=\it \bf f(\it \bf e+x_r,z_{de},\it u)-\it \bf \dot{x}_r
\end{equation}  
is contracting in $\it \bf e$ with a metric $\bf \rm \Theta_{xe}$. The virtual control input $\it \bf z_{de}$ can be regarded as the desired slow manifold for \eqref{reds}, whereas $\it u$ will be decided later. 
This step is analogous to selection of slow component of control law in standard models. Similar to that step, the fast variable $\it \bf z$ is assumed to converge towards $\it \bf z_{de}$ here.  \\
\emph{Step 2:} The virtual control law selected in previous step depends explicitly on control input $u$ which is unknown. The control law $u$ is selected in such a way that the trajectories of $\it \bf z$ subsystem converges to $\it \bf z_{de}$. Define an error variable $e_z=z-z_{de}$ whose dynamics can be expressed as,
\begin{equation}\label{redf}
\mu\it \bf\dot{e}_z=\it \bf g\it \bf(e+x_r,e_z+z_{de},\it u)-\mu \it \bf \dot{z}_{de}.
\end{equation}
The unperturbed part of \eqref{redf} can be expressed as, 
\begin{equation}\label{redfu}
\mu\it \bf\dot{e}_{zu}=\it \bf g\it \bf(e+x_r,e_{zu}+z_{de},\it u)
\end{equation}
Subsequently search for a control law $\it u=\it u(\it \bf e_{zu},e,x_r)$ such that the closed loop dynamics of \eqref{redfu} is contracting in $\it \bf e_{zu}$ with a metric $\rm \bf \Theta_{ez}$. Using lemma \ref{lem1}, the trajectories of the unperturbed system will converge to a unique equilibrium point in the absence of the perturbation term.\\
\emph{\textbf{Remark 4}:} The indirect manifold construction approach utilizes center manifold theory for the selection of $u$ in  \eqref{redf}. Whereas our approach simplify the design procedure by considering the unperturbed part \eqref{redfu} only. Using lemma 2, contraction of \eqref{redfu} implies contraction of \eqref{redf} when the perturbation term $\it \bf \dot z_{de}$ is bounded. \\
Assuming $||\it \bf \dot z_{de}|| \leq d_e$ 
in $\rm \bf B_x \times B_z$. The bound for closed loop system trajectories $\it \bf e_{z}$ is given as:
\begin{equation}\label{bndns1}
||\it \bf e_z(t)-e_{zu}(t)|| \leq \chi_{ze} e^{-\lambda_{ez}t}||\it \bf  e_z(0)-e_{zu}(0)|| + \it \frac{\mu d_e \chi_{ze}}{\lambda_{ez}}
\end{equation}
where $\lambda_{ez}$ is the contraction rate for unperturbed system and $\chi_{ze}$ is the condition number of contraction metric $\rm \bf \Theta_{ze}$.\\
The dynamics of \eqref{err1} can also be written as:
\begin{equation}\label{redsp}
\it \bf \dot{e}_p=\it \bf f(e_p+x_r,z_{de},\it u)-\it \bf \dot{x}_r+ \it \bf f(e_p+x_r,z,\it u)-\it \bf f(e_p+x_r,z_{de},\it u)
\end{equation}
The unperturbed dynamics \eqref{reds} is contracting by suitable selection of $\it \bf z_{de}$ from step 1.  Following the same procedure from theorem 1, the distance between trajectories of \eqref{reds} and \eqref{redsp} satisfy the following bound.
\begin{equation}\label{bndns2}
\lim_{t \to \infty} ||\it \bf e(t)-e_p(t)|| \leq \frac{\mu \it d_z C_1 \chi_x \chi_z}{\it \lambda_x\lambda_{ez}} 
\end{equation}
where $\lambda_{xe}$ is contraction rate of reduced slow system \eqref{reds}, $L_e$ is the Lipschitz constant for $\it \bf f(.)$ in $\it \bf z$ and $\chi_{xe}$ is the condition number of the contraction metric $\rm \bf \Theta_{ex}$. The result is summarized as a theorem below.
\begin{theorem}\label{The3}
The trajectories of closed loop singularly perturbed nonlinear system \eqref{ac1} follow the bounds given in \eqref{bndns1} and \eqref{bndns2}, if the following conditions (i-iii) are met.\\
i) There exists a smooth function $\it \bf z=z_{de}(e,x_r,\dot{x}_r)$ such that the reduced system \eqref{reds} is contracting in $\it \bf e$.\\
ii) There exists a control law $\it u(\it \bf e_z,e,x_r)$ such that the system $\it \bf \dot{e}_{zu}=g(e+x_r,e_{zu}+z_{de},\it u)$ is contracting in $\it \bf e_z$.\\
iii) $||\it \bf \frac{\partial{z_{de}}}{\partial{x}}f(x,z)+\frac{\partial{z_{de}}}{\partial{x_r}}\dot{x}_r|| \leq \it d_e$\\
\end{theorem}
\emph{\textbf{Remark 5}:} The indirect manifold construction based controller design is applicable for both standard and non-standard singularly perturbed problem. However the requirement for non-standard case is more conservative due to the fact that the control law $u$ must ensure contraction of \eqref{redfu} in $\it \bf e_{zu}$ whereas the requirement is partial contraction in standard case.\\

\section{High Gain Scaling for Approximate Feedback Linearizable Systems}
 In the previous sections, we discussed the stabilization of singularly perturbed systems. However, many classes of nonlinear systems which does not inherently possess a time scale separation can forcibly be converted into a singularly perturbed form by high gain feedback. There are certain class of systems for which the assumption of theorem \ref{The1} and \ref{The2} are not needed because of their inherent structure.  Contractive controller design for strict feedback and parametric strict feedback form nonlinear systems are investigated in \cite{sharma1,zamani1}. In this section we take up contractive controller design for approximate feedback linearizable systems which are in the following form. \par
\begin{equation}\label{sys1}
\begin{split}
& \dot{x}=f(x,z)\\
\end{split}
\end{equation}
\begin{equation}\label{rec1}
\begin{split}
& \dot{z}_1 = g_{11}(z_1)+b_1z_2+g_{31}(x,z)\\
& \dot{z}_2 = g_{12}(z_1,z_2)+b_2z_3+g_{32}(x,z)\\
& \dots\\
& \dot{z}_m = g_{1m}(z_1,z_2,z_3....z_m)+b_mu+g_{3m}(x,z)\\
\end{split}
\end{equation}
  $\quad$ where $x \in D_x \subset R^n, z \in D_z \subset R^m$ and  $f(0,0)=0, g_{3j}(0,0)=0$. 
 For $j=1,2 \dots m, g_{1j}:R^j \rightarrow R$ are smooth functions and $b_j$  are positive constants. The sub-system  \eqref{rec1} is in parametric strict feedback form in the absence of 
 $g_3=[g_{31}, g_{32}, \hdots, g_{3m} ]^T$. To derive the results it is assumed that
$||f(0,z)|| \leq c_1||z||$ and $||g_3(0,z)|| \leq c_2||z||$ where $c_1,c_2$ are positive constants.
  This assumption is not conservative in nature and is true for many cases such as flexible link manipulators. The usual backstepping method will  not work for the class of systems considered here due to the presence of $g_3$. The controller design algorithm is divided into two steps. The dynamical systems \eqref{sys1} and \eqref{rec1} are transformed into a singularly perturbed form through a high gain scaling. Then a control law is selected to stabilize the transformed system.
\subsection{Transformation to Singularly Perturbed Form}
 In the absence of $g_3$, \eqref{rec1} is in parametric strict feedback form. Suppose a control law  is selected as \cite{sharma1}: 
 \begin{equation}\label{fbcont1}
u=\frac{1}{b_m}[-g_{1m}(z)+
\sum_{k=1}^{m-1}{\frac{\partial{\alpha_{m}}}{\partial{z_k}}}\dot{z}_k
+u_1]
\end{equation}
\begin{equation}
\begin{split}
&\hat{\xi}_i=z_{i}-\alpha_i,\quad \quad i \in [1,\hdots,m]\\
& \alpha_1=0\\
&\alpha_{i}=\frac{1}{b_{i-1}}[-g_{i-1}(x_1,..x_{i-1})
+\sum_{k=1}^{i-2}{\frac{\partial{\alpha_{i-1}}}{\partial{z_{k}}}}\dot{z}_{k} 
]\quad(\text{for} i \geq 2)\\
\end{split}
\end{equation}
The closed loop system is transformed into a Brunovsky canonical form \eqref{sysbfinal1}.  \par
\begin{equation}\label{sysbfinal1}
\dot{\hat{\xi}}=A\hat{\xi}+Bu_1
\end{equation} 
$A=\begin{bmatrix}&0 &b_1 &0 &\dots &0\\&0 &0 &b_2 &\dots &0\\ \hdotsfor[4]{6}\\&0 &0 &0 & \dots &b_{m-1}\\&0 &0 &0 &\dots &0\\\end{bmatrix},B=\begin{bmatrix}0\\0\\\dots\\0\\1\end{bmatrix}\\$\\
Define a new set of variables scaled by a large positive constant $k$ as:
\begin{equation}
\begin{bmatrix}\eta\\ \dots \\ \xi
\end{bmatrix}=
\begin{bmatrix}k^{m-1}x \\\dots\\ K [\hat{\xi}_1, \hat{\xi}_2,......,\hat{\xi}_{m}]^T
\end{bmatrix}
\end{equation}
 where $K=diag\{k^{m-1}, k^{m-2} ..... k,1\}$. The dynamics  \eqref{sys1} and \eqref{rec1} in terms of new variables are:
\begin{equation}\label{syst}
\dot{\eta}=k^{m-1} f(x,z)=F(\eta,\xi)
\end{equation}
\begin{equation}\label{sysft}
\begin{split}
&\dot{\xi}=K(A\hat{\xi}+Bu_1)+\bar{g}_3(\eta,\xi)\\
&=KAK^{-1}\xi+KBu_1)+\bar{g}_3(\eta,\xi)\\
&=kA\xi+Bu_1+\bar{g}_3(\eta,\xi)\\
\end{split}
\end{equation}
where \[ \bar{g}_3(\eta,\xi)=Kg_3(x,z)|_{x=\frac{\eta}{k^{m-1}}, z_i=\hat{\xi}_{i}+\alpha_i , \hat{\xi}=K^{-1}\xi}.\]
The transformed system can be expressed as a singularly perturbed system in the following form.
\begin{equation}\label{newtr1}
\begin{split}
&\dot{\eta}=F(\eta,\xi)\\
&\mu\dot{\xi}=A\xi+Bu_1+\mu\bar{g}_3(\eta,\xi)
\end{split}
\end{equation}
We will prove that, there exist a control input $u_1$ such that the closed loop system trajectories of \eqref{newtr1} are contracting in ($D_x \times D_z$). The result is stated in the following theorem.
\subsection{Formulation of Control Law}

\emph{\bf{Theorem 4}}: Suppose the following conditions (i-ii) are true for \eqref{syst} and \eqref{sysft} :\par
i)  There exists a function $\xi_1=\rho(\eta)$ such that the system $\dot{\eta}=F(\eta,[\rho(\eta),0, \dots, 0])$ is contracting.\par
ii)  $||\frac{\partial{\bar{g}_3}}{\partial{\xi}}|| \leq \frac{\lambda_{max}[G]}{\mu^2}\ \text{in} \ (D_x\times D_z)$ and $\mu=\frac{1}{k} $ \\
Then there exists a smooth  control law $u$ \\
\begin{equation}\label{ctrl}
\begin{split}
&u=\frac{1}{b_m}[-g_{1m}+\frac{\partial{\alpha_{m}}}{\partial{z_1}}\dot{z}_1+
\sum_{k=1}^{m-1}{\frac{\partial{\alpha_{m}}}{\partial{z_k}}}\dot{z}_k
+u_1]\\
&u_1=k\begin{bmatrix}&-a_1,&-a_2,&\dots&-a_m\end{bmatrix}\xi+ k\ a_1 \ \rho(\eta)\\
\end{split}
\end{equation}
 such that overall closed loop system is contracting.\\
\textbf{Proof:}
 A Hurwitz matrix $G$ is selected as  \\
\begin{equation}
G=\begin{bmatrix}0&b_1&0& \dots & 0\\0&0&b_2& \dots & 0\\ \hdotsfor[4]{5}\\0 &0 &0 &\dots &b_{m-1}\\
-a_1&-a_2&-a_2& \dots & -a_m\end{bmatrix}
\end{equation}
where $a_1,a_2,\dots,a_m$ are suitable positive constants for the stability of $G$. These scalars always exist due to the  companion structure of matrix $G$. Using the control law \eqref{ctrl}, the closed loop fast sub-system of \eqref{newtr1} can be expressed as:
 \begin{equation}\label{sysft2}
\mu\begin{bmatrix}\dot{\xi}_1 \\ \dot{\xi}_2 \\ \dot{\xi}_3 \\ \dots \\
\dot{\xi}_{m-1} \\ \dot{\xi}_{m} \end{bmatrix} = G
\begin{bmatrix}
\xi_1 \\
\xi_2 \\
\xi_3 \\
\dots \\
\xi_{m-1} \\
\xi_{m}\\
\end{bmatrix} + \begin{bmatrix}0\\0\\0\\\dots\\0\\a_1 \rho(\eta)\end{bmatrix}+\mu \begin{bmatrix} \bar{ g}_3 \end{bmatrix}
\end{equation}
 Using Lemma-3, it can be concluded that \eqref{sysft2} is partially contracting in $\xi$ with identity metric if the second condition of theorem 4 is met. Moreover the rate at which \eqref{sysft2} is partially contracting is given by  $\frac{\lambda_{max}(G)}{\mu}-||\mu\frac{\partial{\bar{g}_3}}{\partial{\xi}}||$. Following Lemma 1, the fast subsystem \eqref{sysft2} will converge to a slow manifold given by
\begin{equation}\label{sysp1}
\begin{bmatrix}
&\xi_1 
&\xi_2 
&\hdots 
\xi_{m}
\end{bmatrix}^T  = \begin{bmatrix} &\rho(\eta) & 0  & \hdots &0 \end{bmatrix}^T 
\end{equation}
Substituting $\xi_1=\rho(\eta)$ and $\xi_2,\xi_3,\dots \xi_n=0$ in \eqref{syst}, the reduced slow system becomes 
\begin{equation*}
\dot{\eta}=F(\eta,[\rho(\eta),0,0....0])
\end{equation*}
 From first condition of theorem 4, the reduced system is contracting. $\Box$\\
  \textbf{Remark 5:}
 The slow manifold \eqref{sysp1} is derived by putting $\mu=0$ in \eqref{sysft2}, but for nonzero values of $\mu$ the original slow system will not evolve in same manner as the reduced system. The error bound between the state $\xi_1$ and $\rho(\eta)$ can be derived using Lemma-3. 
Assume the transformed subsystem \eqref{sysft2} is Lipschitz in $\mu$ with a constant $c_4$ and $||\frac{\partial{\rho(\eta)}}{\partial{\eta}}F(\eta,\xi)|| \leq c_5$. Using Theorem 1, the error bound for fast subsystem \eqref{sysft2} can be given as
\begin{equation}\label{stab1}
\begin{split}
&  \lim_{t \to \infty} ||\begin{bmatrix}
\xi_1 \\
\xi_2 \\
\vdots \\
\xi_{m}
\end{bmatrix} - \begin{bmatrix} \rho(\eta) \\ 0  \\ \vdots \\0 \end{bmatrix} || 
 \leq 
  \frac{c_4+c_5}{\frac{\lambda_{max}G}{\mu^2}-||\frac{\partial{\bar{g}_3}}{\partial{\xi}}||} 
\end{split}
\end{equation}
Similarly, error bounds for slow subsystem can also be computed following same steps as theorem 1. Note that
  the stability bounds can be changed according to design goal because $\lambda_{max}(G)$ and $\mu=\frac{1}{k}$ are controller parameters which gives certain amount of relaxation in controller design.\\
 \section{Discussion and Comparison on Simulated Examples}
 \subsection{Stabilization of a D.C Motor}
 Consider an example of d.c motor system described as\cite{nar1}.
 \begin{equation*}\label{dc}
 \begin{split}
 &\dot{x}=-6.39x+6.39z^2\\
 &\mu\dot{z}=-z-[\mu\omega_0]xz+[1+\mu\omega_0]u
 \end{split}
 \end{equation*} 
 where $x$ is the angular velocity, $z$ is the current and $\omega_0=25 \ rad/sec$. The task is to attain a steady state value of one for the slow state $x$. The root of the fast subsystem is given by $z_{ds}=u_1$. The dynamics of reduced system is given by
 \begin{equation*}
 \dot{x}=-6.39x+6.39z_{ds}^2
 \end{equation*}
 For the choice of $z_{ds}=u_1=1$, the reduced system is contracting in $x$. With this choice of $u_1$, the trajectories of the reduced system will converge to an equilibrium point $x=1$. Now $u_2$ is selected as:
 \begin{equation*}
 u_2=\frac{\mu\omega_0}{1+\mu\omega_0}[xz-x]
 \end{equation*}
 It is important to note that $u_2\rightarrow 0$ when $z\rightarrow z_{ds}$. With an overall control law $u=u_1+u_2$, the closed loop fast subsystem is expressed as
 \begin{equation*}
 \mu\dot{z}=-z-\frac{\mu\omega_0}{1+\mu\omega_0}x+[1+\mu\omega_0]
 \end{equation*}
 which is partially contracting in $z$.  In this case $\dot{z}_{ds}=0$ so we do not need to calculate the bound for fourth assumption in theorem 1. The closed loop system is simulated for $\mu=0.1$ which is much larger compared to $\mu=0.02$ obtained from composite controller design in \cite{nar1}. Figure \ref{dcm1} confirms the convergence of slow state $x$ to its desired value one.\par
 \begin{figure}[htbp]
     \centering
      \includegraphics[width=3.25in,height=2.20in]{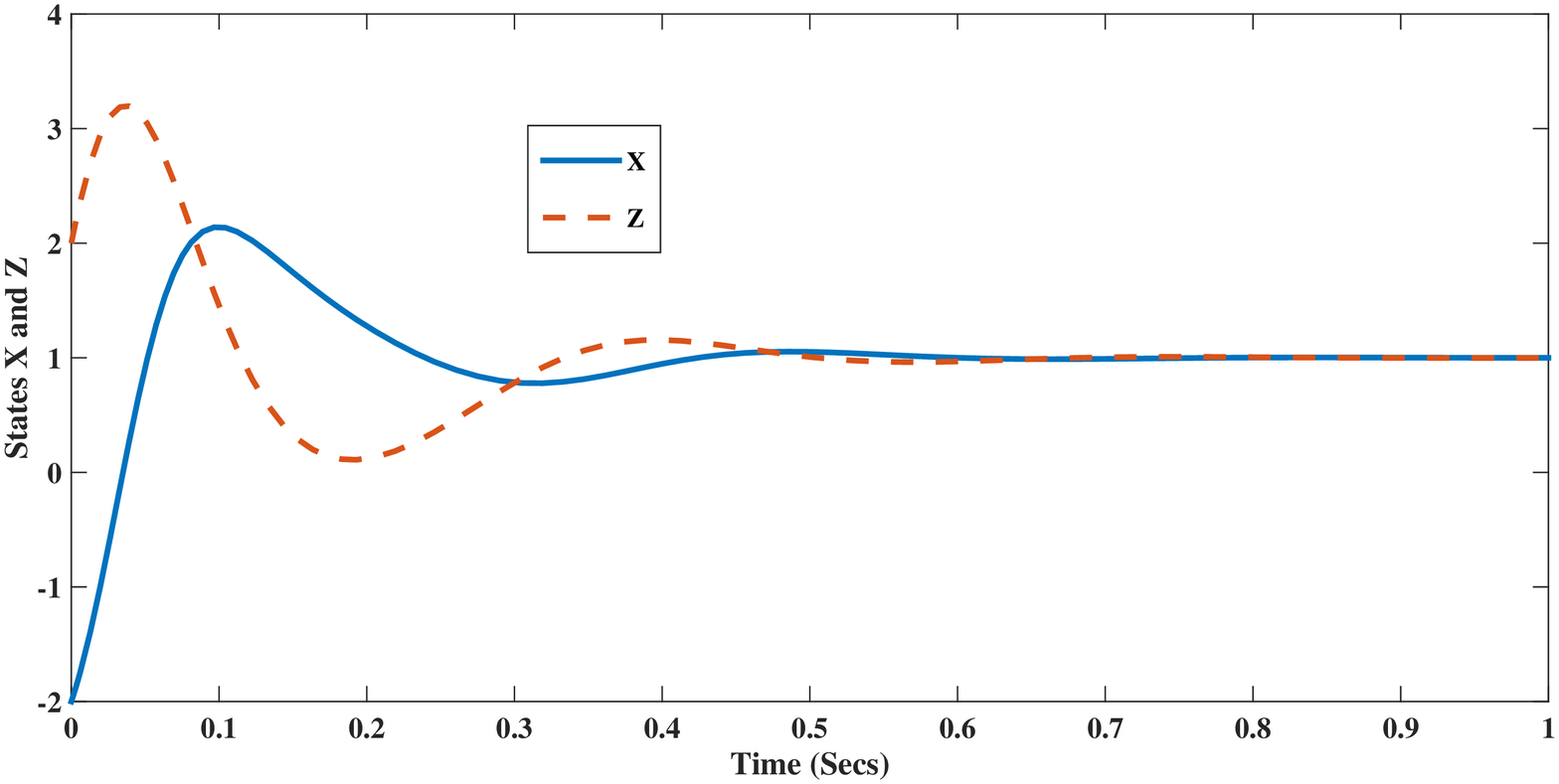}
     \caption{ Closed Loop System Trajectories of D.C Motor for $\mu=0.1$}
      \label{dcm1}
\end{figure} 

\subsection{A Nonstandard Case}
Consider a regulation problem for a nonstandard singularly perturbed system in the form of:
 \begin{equation*}
 \dot{x}=tan (z)-u, 
 \end{equation*}
 \begin{equation*}
 \mu\dot{z}= x +u
 \end{equation*}
 for this problem a choice of $z_{ds}=\tan^{-1} (-x+u)$ will achieve contraction of x-subsystem. Now define $e_z=z-z_{ds}$, the error dynamics can be expressed as
  \begin{equation}\label{ex2f}
 \mu\dot{e}_z= x +u-\mu\dot{z}_{ds}
 \end{equation}
  Select a region $B_z$ in which $\dot{z}_{ds}$ is bounded by some positive constant. If a control law $u=-x-e_z$ is chosen, then the unperturbed part of \eqref{ex2f} is contracting in $B_z$. From theorem 3, $e_z$ will converge to a small neighborhood around origin. Approximating the limiting value of $e_z \rightarrow 0$, the slow manifold can be expressed as $z_{ds}=\tan^{-1} (-2x)$. The control law to be implemented is given by $u=-x-z+z_{ds}$ which is given by:
 \begin{equation*}
 u=-x-z+\tan^{-1} (-2x)
 \end{equation*}
 The simulation results for this control law with a perturbation parameter $\mu=0.2$ is given in Figure \ref{Example ns}.\\
 \begin{figure}[htbp]
     \centering
      \includegraphics[width=3.25in,height=2.20in]{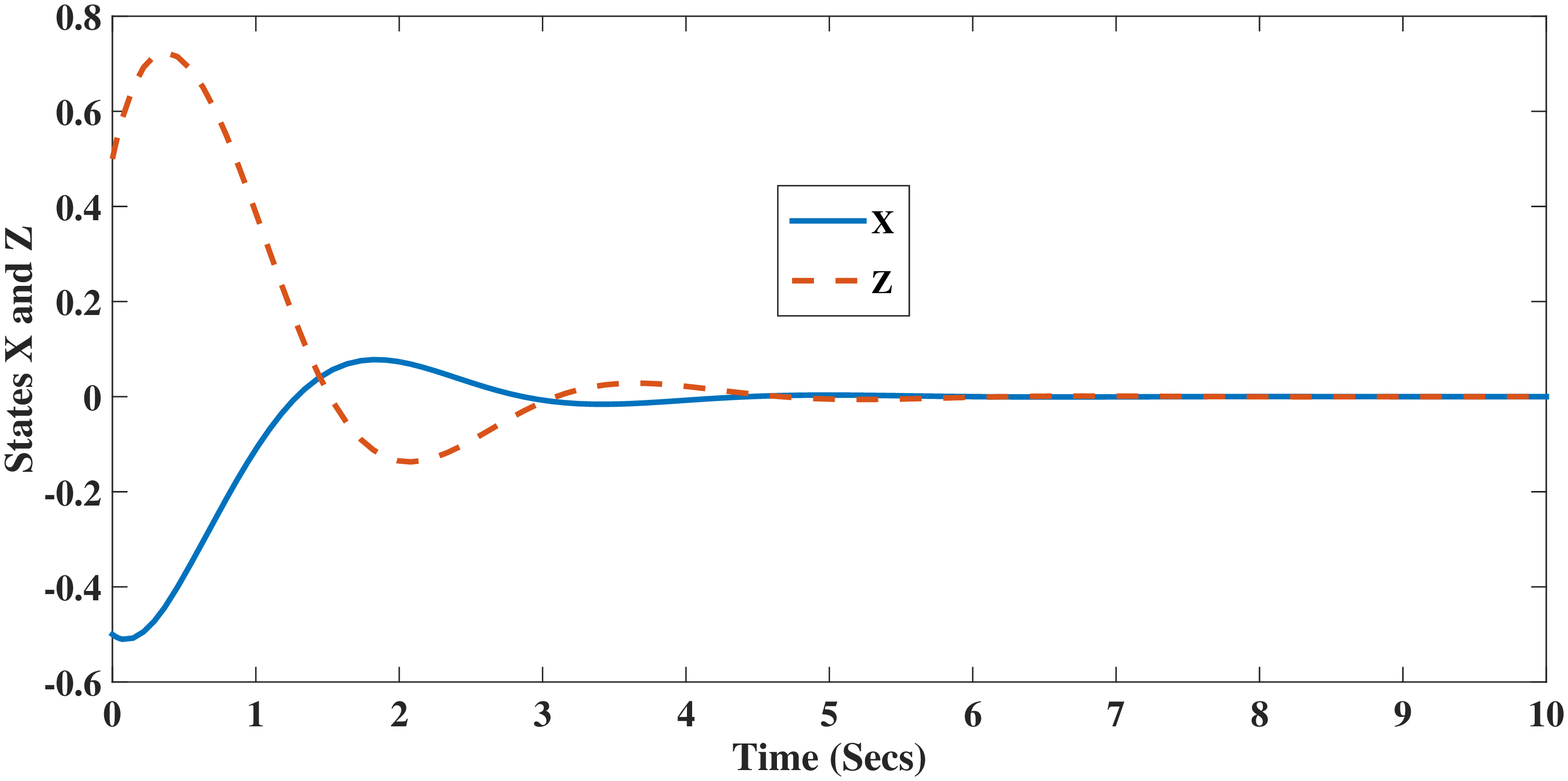}
      \caption{ Convergence of System Trajectories for $\mu=0.2$}
      \label{Example ns}
\end{figure} 
From this figure, it can be concluded that the control law proposed in theorem 3 is able to provide bounded tracking performance for nonstandard singularly perturbed systems. Although the magnitude of $\mu$ has a direct effect on the magnitude of steady state error, the convergence is guaranteed.\\
 \subsection{High Gain Scaling}
 \begin{equation}\label{ex1}
 \begin{split}
 & \dot{x}=x^2+z_1+xz_2\\
 & \dot{z}_1=x sin z_2+ z_2\\
 & \dot{z}_2=u
 \end{split}
 \end{equation}
 Applying recursive control design for strict feedback part(z-subsystem) and defining new variables as $\alpha_1=0,\hat{\xi}_1=z_1,\alpha_2(z_1)=-z_1,\hat{\xi}_2=z_2-\alpha_2$, the dynamics becomes
 \begin{equation*} \begin{bmatrix}\dot{\hat{\xi}}_1\\\dot{\hat{\xi}}_2\end{bmatrix}=\begin{bmatrix}\hat{\xi}_1\\u+z_2
 \end{bmatrix}+\begin{bmatrix}xsinz_2\\0\end{bmatrix}
 \end{equation*}
Define a new transformation,
 \begin{equation*}
\begin{bmatrix}\eta\\   \xi
\end{bmatrix}=
\begin{bmatrix}k x \\ K [\hat{\xi}_1, \hat{\xi}_2]
\end{bmatrix}
\end{equation*}
 where $K=diag\{ k,1\}$. System \eqref{ex1} is transformed to
\begin{equation} \begin{bmatrix}\dot{\eta}\\\dot{\xi}_1\\\dot{\xi}_2\end{bmatrix}=
\begin{bmatrix}\frac{1}{k}\eta^2+\xi_1+\eta\xi_2\\k\xi_2\\u\end{bmatrix}+\begin{bmatrix}
0\\\eta sin(\xi_2)\\0\end{bmatrix}
\end{equation}
To design the control law $u$, choose $a_1=a_2=2$ and $\rho(\eta)=(\frac{-1}{k}\eta^2-\eta)$ in \eqref{ctrl} such that
\begin{equation}\label{uex1}
\begin{split}
&u=-2k\xi_2-2k\xi_1-2k(\frac{-1}{k}\eta^2-\eta)\\
\end{split}
\end{equation}
 Denote $\mu=\frac{1}{k}$ and the fast subsystem of closed loop system reduces to
\begin{equation}\label{transex1}
\begin{split}
& \mu \begin{bmatrix}\dot{\xi}_1\\\dot{\xi}_2\end{bmatrix}=
\begin{bmatrix}\xi_2\\(-2\xi_1-\xi_2-2(\frac{-1}{k}\eta^2-\eta))\end{bmatrix}+\mu \begin{bmatrix}
\eta sin(\xi_2)\\0\end{bmatrix}
\end{split}
\end{equation}
 System \eqref{transex1}  is partially contracting in a set $S \subset R^n$, if $\mu||\frac{\partial}{\partial{\xi_2}}{\eta sin(\xi_2)}|| \leq ||\frac{\lambda_{max}G}{\mu}||$ in $S$. Therefore the slow manifold can be described as,
\begin{equation}
\begin{split}
& \begin{bmatrix}\xi_1\\\xi_2\end{bmatrix}=
\begin{bmatrix}\frac{-1}{k}\eta^2-\eta\\0\end{bmatrix}\\
\end{split}
\end{equation}
Substituting the value of $\xi_1$, the slow subsystem $ \dot{\eta}=\frac{1}{k}\eta^2+\xi_1+\eta\xi_2 $  reduces to $\dot{\eta}=-\eta$ which is also contracting. The simulations for the closed loop system is illustrated in figure 4 and 5. A gain of $k=10$ and initial conditions $[-1,1,0]$ are used for simulation.  The peaking phenomenon observed in control law is due to the high gain feedback. It can be reduced by using a saturater of desirable magnitude.\\
\emph{\textbf{Remark 6}:}The stability guaranteed for a broad range of $\mu \in [0,1]$ rather than a restrictive maximum value. This perspective gives more freedom in the choice of the high gain parameter as $\mu=\frac{1}{k}$.  In other words, the controller can achieve stabilization of the closed loop system using less control effort. However the contraction rate and error bounds will be different for different values of $k$. The closed loop system trajectories is shown in Figure \ref{Example 1 ek} for $k=4.5$. Another advantage of contraction based control design is that, the bound of error \eqref{stab1} between fast subsystem and slow manifold, ($\xi_1-(\rho(\eta))$)  can be changed by changing the parameters of control law \eqref{uex1}. The choice of matrix $G$ has a direct effect on the stability bounds \eqref{stab1} because its maximum eigen value decides the steady state error for the closed loop system. 
\begin{figure}[htbp]
     \centering
      \includegraphics[width=3.25in,height=2.20in]{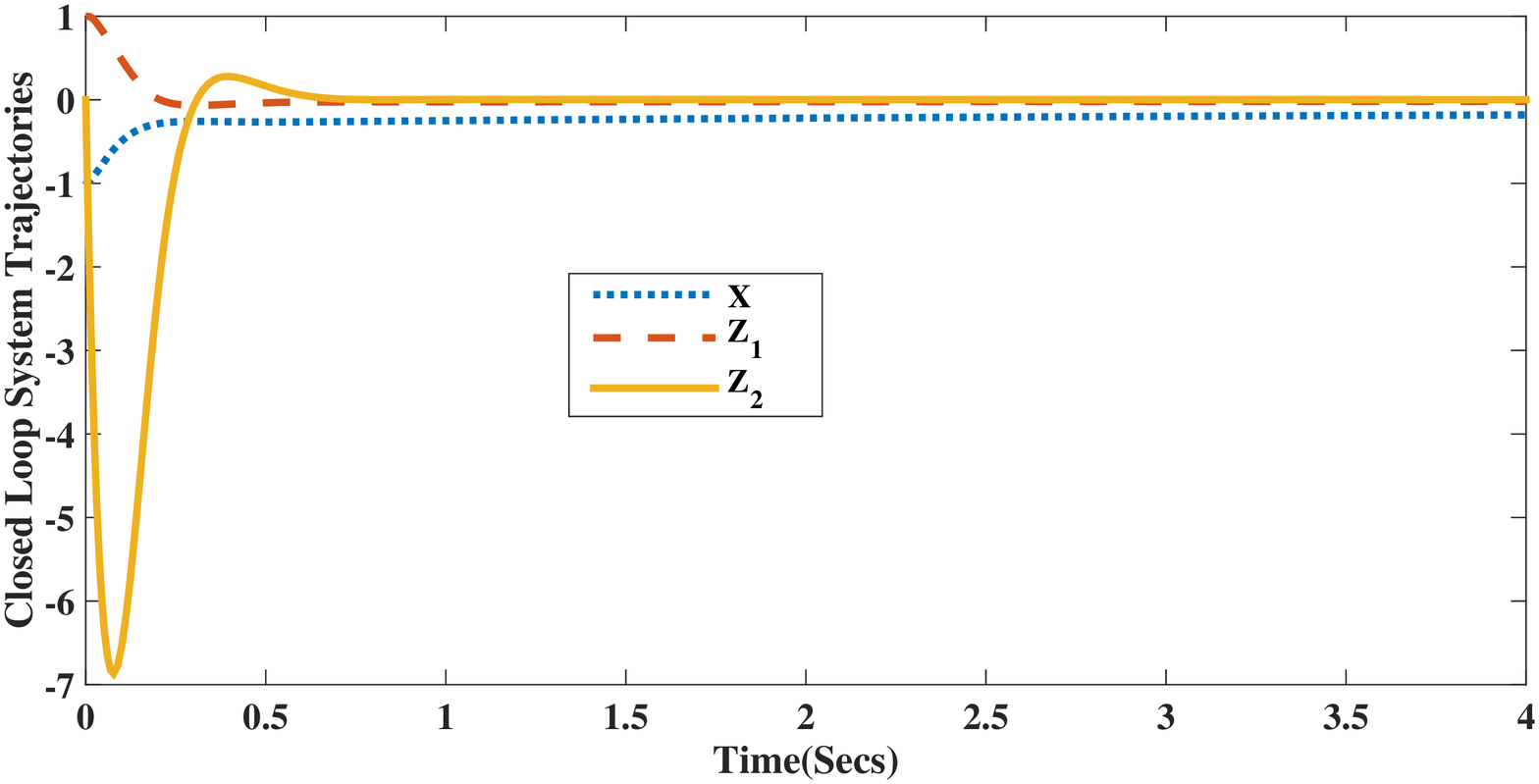}
      \caption{ Closed Loop System for k=10}
      \label{Example 1}
\end{figure} 
\begin{figure}[htbp]
     \centering
   \includegraphics[width=3.25in,height=2.20in]{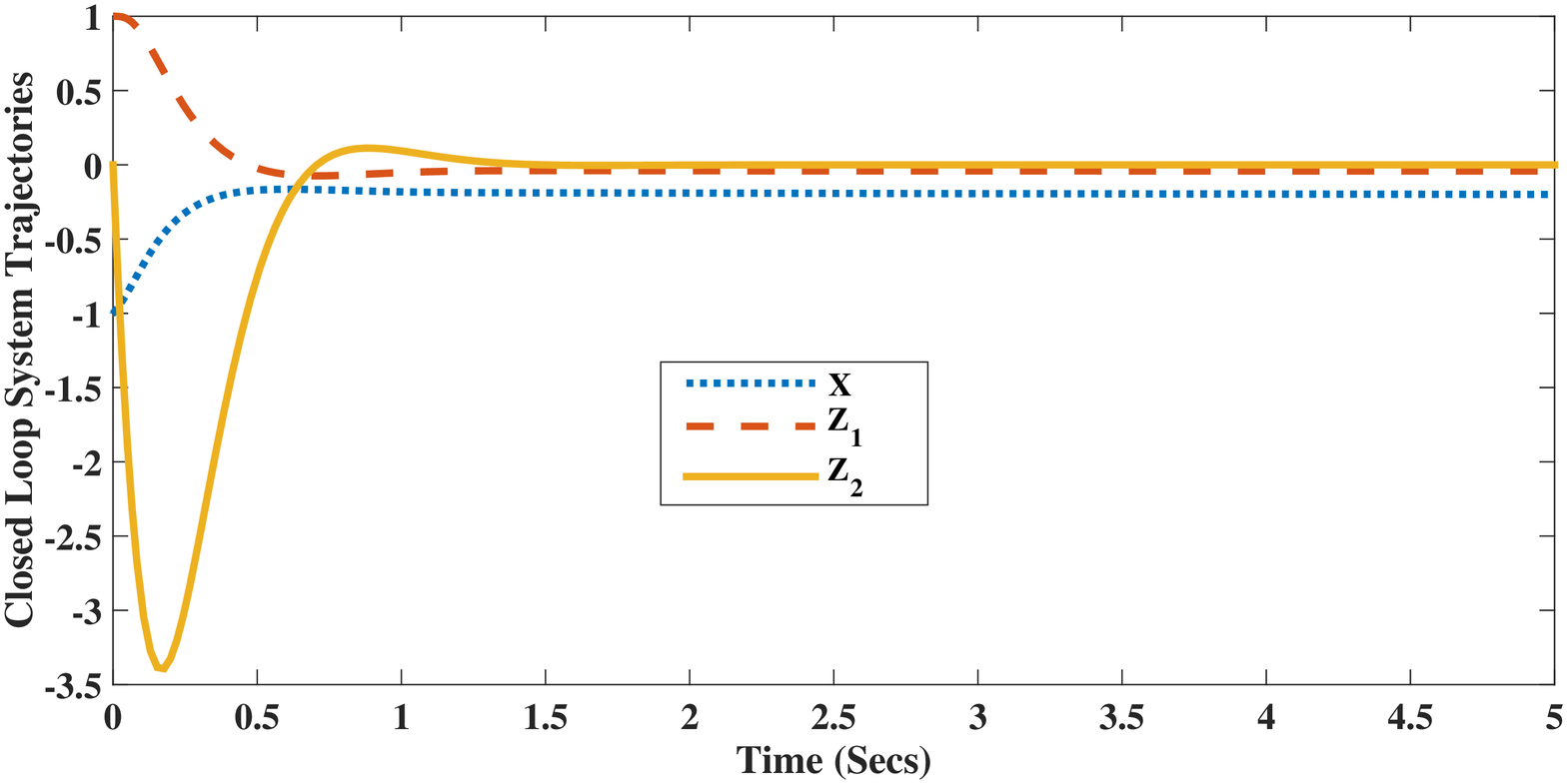}
      \caption{ Closed Loop System for k=4.5}
      \label{Example 1 ek}
   \end{figure} \\
\section{CONCLUSIONS}
A new approach for stabilization of singularly perturbed system is formulated based on contraction theory. The controller design formalism does not require any interconnection conditions. The trajectories of the closed loop system converge to an ultimate bound irrespective of the magnitude of perturbation parameter. Moreover an exponential convergence of trajectories  can also  be achieved under certain restrictions.  The proposed design framework is extended to develop a high gain based control law for approximate feedback linearizable systems. The design methodology presented here can assure ultimate boundedness of trajectories even if the Lyapunov based bound on perturbation parameter is breached due to some design constraints. The methodology presented here provides some relaxation in the choice of high gain parameter and can be useful to high gain observers  for nonlinear systems.


\bibliographystyle{cite}

\bigskip

\end{document}